\documentclass{article}

\usepackage{verbatim}
\usepackage{graphics}
\usepackage[pdftex]{graphicx}
\usepackage{amsmath, amssymb, amsfonts, amsthm}
\usepackage{url}
\usepackage{setspace}
\usepackage{algorithmic}
\newtheorem{theorem}{Theorem}

\newtheorem{lemma}[theorem]{Lemma}

\newtheorem{open}{Open Problem}
\begin{document}

\author{Ben D. Lund\footnote{lund.ben@gmail.com. Department of Computer Science, University of Cincinnati, Cincinnati, OH 45221, USA.}
\and George B. Purdy\footnote{george.purdy@uc.edu. Department of Computer Science, University of Cincinnati, Cincinnati, OH 45221, USA.}
\and Justin W. Smith\footnote{smith5jw@mail.uc.edu. Department of Computer Science, University of Cincinnati, Cincinnati, OH 45221, USA.}
\and Csaba D. T\'oth\footnote{cdtoth@ucalgary.ca. Department of Mathematics and Statistics, University of Calgary, Calgary, AB, Canada.}}
\title{Collinearities in Kinetic Point Sets}
\maketitle

\begin{abstract}
Let $P$ be a set of $n$ points in the plane, each point moving along a given trajectory. A {\em $k$-collinearity} is a pair $(L,t)$ of a line $L$ and a time $t$ such that $L$ contains at least $k$ points at time $t$, the points along $L$ do not all coincide, and not all of them are collinear at all times. We show that, if the points move with constant velocity, then the number of 3-collinearities is at most $2\binom{n}{3}$, and this bound is tight. There are $n$ points having $\Omega(n^3/k^4 + n^2/k^2)$ distinct $k$-collinearities. Thus, the number of $k$-collinearities among $n$ points, for constant $k$, is $O(n^3)$, and this bound is asymptotically tight. In addition, there are $n$ points, moving in pairwise distinct directions with different speeds, such that no three points are ever collinear.
\end{abstract}

\section{Introduction}

Geometric computation of moving objects is often supported by kinetic data structures (KDS), introduced by Basch, Guibas and Hershberger~\cite{Basch99,guibas04}. The combinatorial structure of a configuration is described by a set of certificates, each of which is an algebraic relation over a constant number of points. The data structure is updated only if a certificates fails. A key parameter of a KDS is the maximum total number of certificate failures over all possible simple motions of $n$ objects. For typical tessellations ({\em e.g.}, triangulations~\cite{KaplanRS11} or pseudo-triangulation~\cite{KirkpatrickS02}) or moving points in the plane, a basic certificate is the orientation of a triple of points, which changes only if the three points are collinear.

We are interested in the maximum and minimum number of collinearities among $n$ kinetic points in the plane, each of which moves with constant velocity. A {\em $k$-collinearity} is a pair $(L,t)$ of a line $L$ and a time $t$ such that $L$ contains at least $k$ points at time $t$, the points along $L$ do not all coincide, and not all of them are collinear at all times. The last two conditions help to discard a continuum of trivial collinearities: we are not interested in $k$ points that coincide, or are always collinear ({\em e.g.} if they move with the same velocity).

\smallskip\noindent{\bf Results.}
The {\em maximum} number of 3-collinearities among $n$ kinetic points in the plane, each moving with constant velocity, is $2\binom{n}{3}$. In particular, if three points are not always collinear, then they become collinear at most twice. Moreover, the maximum is attained for a kinetic point set where no three points are always collinear.
We also show that, for constant $k$, the number of $k$-collinearities is $O(n^3)$, and this bound is asymptotically tight. In the lower bound construction, $\Omega(n^3/k^4 + n^2/k^2)$ $k$-collinearities occur such that at each  $k$-collinearity at most $\lceil k/2\rceil$ of the points are always collinear.

The {\em minimum} number of collinearities among $n$ kinetic points in the plane is obviously 0. Consider, for example, $n$ points in general position that have the same velocity. We construct $n$ kinetic points that move with pairwise distinct speeds in different directions, and yet they admit no 3-collinearities.

We assume an infinite time frame $(-\infty,\infty)$. The motion of a point $p$ in $\mathbb{R}^d$ can be represented by its trajectory in $\mathbb{R}^{d+1}$, where the last (``vertical'') dimension is time. If a point $p$ moves with constant velocity in $\mathbb{R}^d$, its trajectory is a nonhorizontal line $L_p\subset \mathbb{R}^{d+1}$. Every algebraic condition on kinetic points in $\mathbb{R}^d$  has an equivalent formulation in terms of their trajectories in $\mathbb{R}^{d+1}$. We use both representations throughout this paper.

\smallskip\noindent{\bf Related previous results.}
Previous research primarily focused on collisions. Two kinetic points $p,q\in \mathbb{R}^d$ collide if and only if their trajectories $L_p,L_q\subset \mathbb{R}^{d+1}$ intersect. A {\em $k$-collision} is a pair $(P,t)$ of a point $P\in \mathbb{R}^d$ and a time $t$ such that at least $k$ kinetic points meet at $P$ at time $t$, but not all these points are always coincident. It is easy to see that for $n$ points in $\mathbf{R}^1$, each moving with constant velocity, the number of 2-collisions is at most $\binom{n}{2}$, and this bound is tight.
The number of $k$-collisions in $\mathbf{R}^1$ is $O(n^2/k^3 + n/k)$, and this bound is also the best possible, due to the Szemer\'edi-Trotter theorem \cite{Szemeredi83}.

Without additional constraints, the bounds for the number of collisions remains the same in $\mathbb{R}^d$ for every $d\geq 1$, since the points may be collinear at all times. Sharir and Feldman~\cite{Sharir94,FeldmanS05} considered the number of 3-collisions in the plane among points that are not always collinear. The trajectories of such a 3-collision form a so-called ``joint'' in 3-space. Formally, in an arrangement of $n$ lines in $\mathbb{R}^{d+1}$, a {\em joint} is a point incident to at least $d+1$ lines, not all of which lie in a hyperplane. Recently, Guth and Katz~\cite{Guth10} proved that $n$ lines in $\mathbb{R}^3$ determine $O(n^{3/2})$ joints. Their proof was later  generalized and simplified~\cite{ElekesKS11,KaplanSS10}: $n$ lines in $\mathbb{R}^{d+1}$ determine $O(n^{(d+1)/d})$ joints. These bounds are the best possible, since $\Theta(n^{(d+1)/d})$ joints can be realized by $n$ axis-parallel lines arranged in a grid-like fashion in $\mathbb{R}^d$. However no nontrivial bound is known for the number of joints under the additional constraint that no $d$ lines lie in a hyperplane.

A $k$-collinearity is the natural generalization of a $k$-collision in dimensions $d\geq 2$. It is easy to give a $\Theta(n^3)$ bound on the maximum number of 3-collinearities in the plane, since three random points, with random velocities, form $\Theta(1)$ collinearities in expectation. However, a 4-collinearity assumes an algebraic constraint on the trajectories of the 4 kinetic points. Here we present initial results about a new concept, including tight bounds on the number of 3-collinearities in the plane, and asymptotically tight bounds on the number of $k$-collinearities in the plane, for constant $k$.

\smallskip\noindent{\bf Organization.} We present our results for the maximum number of 3- and $k$-collinearities in Section~\ref{sec:upper}. We construct a kinetic point set with no collinearities in Section~\ref{sec:nocollin} and conclude with open problems in Section~\ref{sec:conclude}.

\section{Upper bound for 3-collinearities}\label{sec:upper}

Given any two kinetic points $a$ and $b$ in the plane, denote by $S_{a, b}$ the set of point-time pairs in $\mathbb{R}^3$ that form a 3-collinearity with $a$ and $b$.
This will be the set of all horizontal lines that intersect both $L_a$ and $L_b$.
We can find the times at which a third point, $c$, is collinear with $a$ and $b$ by characterizing the set $L_c \cap S_{a, b}$.
In particular, the cardinality of $L_c \cap S_{a,b}$ is the number of 3-collinearities formed by these three points.

The first issue is to characterize the set $S_{a, b}$.
For this purpose, we will use a classical geometric result.

\begin{lemma}[14.4.6 from \cite{berger1987geometry}]\label{lem:skew}
Let $L_a$ and $L_b$ be disjoint lines in a three-dimensional Euclidean affine space, and let $a$ and $b$ be points moving along $L_a$ and $L_b$ with constant speed. The affine line through $a$ and $b$ describes a hyperbolic paraboloid as $t$ ranges from $-\infty$ to $\infty$.
\end{lemma}

This is a special case of a construction that produces a hyperboloid of one sheet or a hyperbolic paraboloid from three skew lines \cite[p. 15]{hilbert1952geometry}.
Given three skew lines, the union of all lines that intersect all three given lines is a doubly ruled surface.
If the three given lines are all parallel to some plane, the surface will be a hyperbolic paraboloid; otherwise, the surface will be a hyperboloid of a single sheet.

Given two kinetic points $a$ and $b$ moving at constant velocity, we can arbitrarily choose three horizontal lines that intersect $L_a$ and $L_b$ to use with the above construction.
Since horizontal lines are parallel to a horizontal plane, the resulting surface will be a hyperbolic paraboloid.

This characterizes $S_{a, b}$ in the case that $L_a$ and $L_b$ are skew.
It remains to extend the characterization to the cases that $a$ and $b$ collide or have the same speed and direction.

\begin{lemma}\label{lem:intersectionSurface}
Given two kinetic points, $a$ and $b$, each moving with constant velocity, there are three possibilities for $S_{a, b}$.
\begin{enumerate}
\item If $a$ and $b$ have the same direction and speed, then $S_{a,b}$ is a non-horizontal plane.
\item If $a$ and $b$ collide, then $S_{a,b}$ is the union of a horizontal and a non-horizontal plane.
\item Otherwise, $S_{a,b}$ is a hyperbolic paraboloid.
\end{enumerate}
\end{lemma}
\begin{proof}
If $L_a$ and $L_b$ intersect or are parallel, then there is a unique plane $\Pi$ that contains both $L_a$ and $L_b$.
Since neither $L_a$ nor $L_b$ is horizontal, $\Pi$ is not horizontal. Every point in $\Pi$ belongs to the union of all horizontal lines containing a point from each of $L_a$ and $L_b$.

Since two non-coincident points span a unique line and the intersection of $\Pi$ with a horizontal plane is a line, if $L_a$ and $L_b$ are parallel, then $S_{a, b} = \Pi$.
This covers the case that $a$ and $b$ have the same direction and speed.

If $L_a$ and $L_b$ intersect, then every point in the horizontal plane $\Pi'$ containing the intersection point $L_a\cap L_b$ is on a horizontal line containing a point from each of $L_a$ and $L_b$. In this case, $S_{a,b} = \Pi \cup \Pi'$. This covers the case that $a$ and $b$ collide.

If $L_a$ and $L_b$ are skew, Lemma \ref{lem:skew} implies that $S_{a,b}$ is a hyperbolic paraboloid.
This covers the generic case.
\end{proof}

\begin{lemma}
Three points in the plane, each moving with constant velocity, will either be always collinear or collinear at no more than two distinct times.
\end{lemma}
\begin{proof}

Label the points $a$, $b$, and $c$.
By lemma \ref{lem:intersectionSurface}, $S_{a,b}$  is a plane, the union of two planes, or a hyperbolic paraboloid.
Every time $L_c$ intersects $S_{a,b}$, the points $a$, $b$, and $c$ are collinear. Since a plane is a surface of degree $1$ and a hyperbolic paraboloid is a surface of degree $2$, $L_c$ cannot intersect $S_{a,b}$ more than twice without being contained in $S_{a,b}$.
\end{proof}

\begin{theorem}\label{th:threePointLineUpperBound}
A set of $n$ points in the plane, each moving with constant speed and direction,
determines no more than $2\binom{n}{3}$ 3-collinearities.
\end{theorem}
\begin{proof}
There are $\binom{n}{3}$ subsets of $3$ points, each of which forms at most two 3-collinearities.
\end{proof}

Clearly, this bound applies directly to $k$-collinearities, for any $k\geq 3$.
If no three points are always collinear, this bound can easily be improved for $k > 3$.

\begin{theorem}\label{th:kPointLineUpperBound}
A set of $n$ points in the plane, each moving with constant speed and direction,
and no three of which are always collinear, determines no more than
$2\binom{n}{3}/\binom{k}{3}$ $k$-collinearities.
\end{theorem}
\begin{proof}
By Theorem \ref{th:threePointLineUpperBound}, there are at most $2 \binom{n}{3}$ sets of $3$ instantaneously collinear points. A $k$-collinearity accounts for at least $\binom{k}{3}$ distinct sets of $3$ instantaneously collinear points.
\end{proof}

\subsection{The $2 \binom{n}{3}$ bound is tight for 3-collinearities}
\begin{theorem}\label{th:tight}
Theorem \ref{th:threePointLineUpperBound} is tight for the case $k=3$.
\end{theorem}
\begin{proof}
We construct a set of $n$ kinetic points, no three always collinear, such that they admit exactly $2\binom{n}{3}$ 3-collinearities.

Let the points be $\{p_1, p_2, \ldots, p_n\}$. Each point moves with speed $1$.
The direction of motion of point $p_i$ forms an angle of $\theta_i = 3\pi/2 + \pi/(4i)$ with the positive $x$ direction.
At time $t=0$, each point is on a circle of radius $1$ centered at $(-1,1)$, and positioned so that its direction of travel will cause it to cross the origin at some later time. Since two locations on the circle might satisfy this property, we choose the one closer to the origin (Fig. \ref{fig:manyLines0}).

\begin{figure}[h]
\center{
\includegraphics[height=2.5in,width=2.5in]{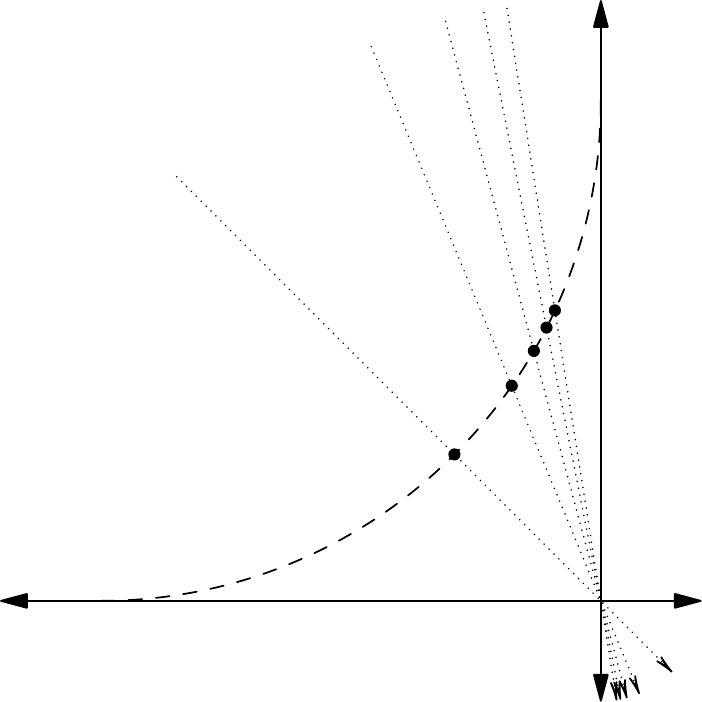}}
\caption{\label{fig:manyLines0} A set of kinetic points forming $2 \binom{n}{3}$ three point lines over the time interval $(-\infty,\infty)$, at time $0$.}
\end{figure}

At time $t=0$, no three points are collinear, so no triple of points is always collinear.
Choose any three elements from $\{p_1, p_2, ...,p_n\}$, say $p_j$, $p_k$, and $p_l$, such that
$\theta_j < \theta_k < \theta_l$. We will show that these points are collinear at two distinct times.

Let $H_R$ and $H_R$ denote the left and right halfplanes, respectively, determined by the directed line $p_jp_l$.
Let $C$ be a closed convex curve passing through $p_j$,  $p_k$, and $p_l$ such that it crosses line $p_jp_l$ at $p_j$ and $p_l$ only. We can determine which half-plane contains $p_k$ from the cyclic order of the three points on $C$. If the clockwise order is $(p_j, p_k, p_l)$, then $p_k \in H_L$; if the clockwise order is $(p_j, p_l, p_k)$, then $p_k \in H_R$.

At time $0$, the points are distributed on the circle of radius $1$ with center $(-1,1)$, and the clockwise order of the chosen points on this circle is $(p_j, p_k, p_l)$. Thus, $p_k \in H_L$.

Let $c_i$ be the distance between $p_i$ and the origin at time $0$.
Since all points are initially moving toward the origin at a speed of $1$, the distance between $p_i$ and the origin is $|c_i-t|$ at time $t$.

We now establish that $p_k$ is in $H_R$ for $|t| \gg 1$. If $t \gg 1$, all of the points $\{p_1, p_2, ...,p_n\}$ will lie approximately on a circle of radius $t$ centered at the origin. The clockwise order of the points on a convex curve approximating this circle will be $(p_j, p_l, p_k)$, and $p_k \in H_R$. Likewise, when $t \ll -1$ the points will be approximately on a circle of radius $|t|$ (but at points antipodal to those when $t \gg 1$), and the order will be $(p_j, p_l, p_k)$ with $p_k \in H_R$. Figure~\ref{fig:manyLinesInf} depicts the configuration for $t \gg 1$.

\begin{figure}[h]
\center{
\includegraphics[height=2.5in,width=2.5in]{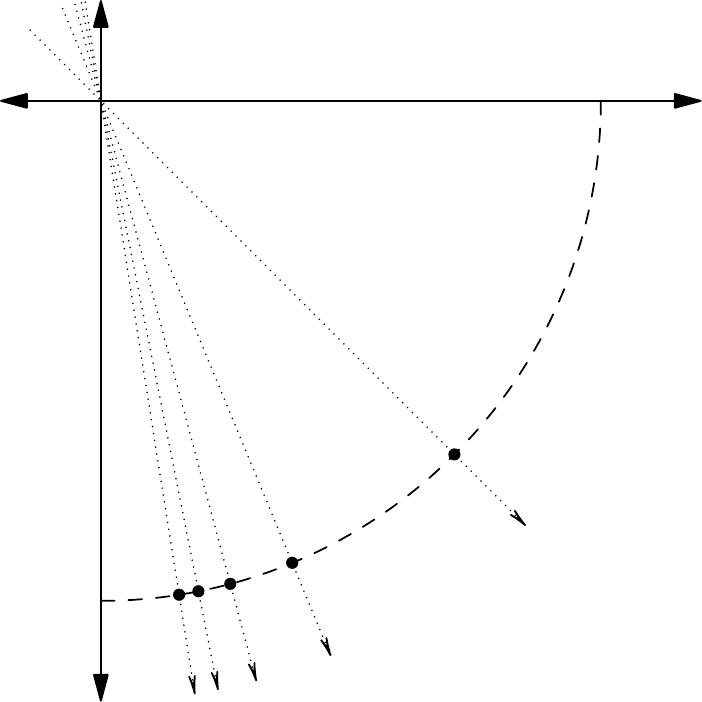}}
\caption{\label{fig:manyLinesInf} A set of kinetic points forming $2 \binom{n}{3}$ three point lines over the time interval $(-\infty,\infty)$, at time $\gg 1$.}
\end{figure}

Since $p_k$ alternates from $H_L$ to $H_R$ and back to $H_R$ as $t$ goes from negative to positive infinity,
there must exist times $t'$ and $t''$ at which the three points are collinear.
\end{proof}

The above construction is degenerate in the sense that the paths of the points are all concurrent through the origin.
Note that our argument is not sensitive to a small perturbation in the location or the direction of the points. The direction of motion of each point may be perturbed so that the trajectories are in general position.

Additionally, the construction may be altered so that the points travel at different speeds. If the speeds of $\{p_1, p_2, ..., p_n\}$ are not all the same, then the points will not approach a circle as $|t|$ approaches $\infty$.
However, as long as no three points are always collinear and the points approach some closed convex curve as $|t|$ approaches infinity, the arguments used will remain valid. For example, if the speed of point $p_i$ is $1/(1 - \cos(\theta_i)/2)$, then for $|t| \gg 1$, the points will be approximately distributed on an ellipse enclosing the origin. This ensures that any three points will be collinear at two distinct times,
so the set of $n$ points will have $2\binom{n}{3}$ 3-collinearities.

\subsection{The $O(n^3)$ bound is tight for fixed $k$}

By Theorem \ref{th:threePointLineUpperBound}, $n$ kinetic points moving with constant velocities determine $O(n^3)$ $k$-collinearities. Here for all integers $n\geq k\geq 3$, we construct a set of $n$ kinetic points that determines $\Omega(n^3/k^4 + n^2/k^2)$ $k$-collinearities.

First assume that $n\geq k^2$. We construct $n$ kinetic points with $\Omega(n^3/k^4)$ $k$-collinearities. The points will move on two parallel lines $L_1:x=0$ and $L_2:x=1$ in varying speeds, and a simultaneous $\lfloor k/2\rfloor$-collision on $L_1$ and a $\lceil k/2\rceil$-collision on $L_2$ defines a $k$-collinearity.

Without loss of generality we may assume that $n$ is a multiple of $k$.
Let $\{A_1, A_2,..., A_{\lfloor k/2 \rfloor}\}$ and $\{B_1, B_2, ..., B_{\lceil k/2 \rceil}\}$ be sets of $n/k$ points each. At time $0$, let $A_i = \{a_{i,j} = (0,j) :  j= 1,...,n/k\}$ for $1 \leq i \leq \lfloor k/2 \rfloor$, and let $B_i = \{b_{i,j} = (1,j) : j = 1,...,n/k\}$ for $1 \leq j \leq \lceil k/2 \rceil$. All point move in the direction $(0,1)$. The points in $A=\bigcup_{i=1}^{\lfloor k/2\rfloor}A_i$ are always in line $x=0$, and the point in $B=\bigcup_{i=1}^{\lceil k/2\rceil} B_i$ are always in line $x=1$. The speed of a point in set $A_i$ or $B_i$ is $i-1$; for example, each point in set $A_1$ has speed $0$.

At each time $t = \{0, 1, ..., n/(k\lfloor k/2 \rfloor)\}$, there are $(n/k - (k - 1)t)$ $\lfloor k/2 \rfloor$-way collisions among points in $A$ and $(n/k - (k-1)t)$ $\lceil k/2 \rceil$-way collisions among points in $B$.
Each line connecting a $\lfloor k/2 \rfloor$-collision among points in $A$ and a $\lceil k/2 \rceil$-collision among points in $B$ is a $k$-collinearity. Thus, at each time $t=\{0,1,...,n/(k\lfloor k/2 \rfloor) \}$, there are $(n/k - (k-1)t)^2$ k-collinearities. Taking the sum, the number of $k$-collinearities over $t=[0, \infty)$ is
\begin{eqnarray*}
\sum_{t = 0}^{n/(k\lfloor k/2 \rfloor)} (n/k - (k-1)t)^2 & \geq & \sum_{t = 0}^{n/(k\lfloor k/2 \rfloor)}(k-1)^2t^2 \\
& \geq & (k-1)^2 \sum_{t=0}^{n/(k\lfloor k/2 \rfloor)}t^2 \\
&=& \Omega(n^3/k^4).
\end{eqnarray*}

Now assume that $k\leq n<k^2$. We construct $n$ kinetic points with $\Omega(n^2/k^2)$ $k$-collinearities.
The $n$ points are partitioned into subsets, $A_1,A_2,\ldots , A_{\lfloor n/k\rfloor}$, each of size at least $\lceil k/2\rceil$. The points in each subset have a single $\lceil k/2\rceil $-collision at time $0$, at points in general position in the plane. Any line between two $\lceil k/2\rceil$-collisions is a $k$-collinearity. Hence there are $k$-collinearities is $\Omega(n^2/k^2)$.

\section{Kinetic point sets with no collinearities}\label{sec:nocollin}
It is clearly possible to have no 3-collinearities among $n$ kinetic points if the points move with the same direction and speed---this is simply a set of relatively static points, no three of which are collinear. Similarly, if we are only interested in collinearities in the time interval $(0,\infty)$, it is clearly possible to have no collinearities---any set of kinetic points will have a final 3-collinearity.

Less obviously, we can construct $n$ kinetic points, any two of which have different direction and speed,
that admit no 3-collinearities over the time interval $(-\infty, \infty)$.
\begin{theorem}
For every integer $n\geq 1$, there is a set of $n$ points in the plane, each moving with constant speed and direction, no two of the points having the same speed or direction, such that no three points are collinear over the time interval $(-\infty, \infty)$.
\end{theorem}
\begin{proof}
We will start by constructing a set of kinetic points with no 3-collinearities, having different directions but the same speed. Then, we will modify the construction so that the points move with different speeds.

For $1 \leq i \leq n$, let $\theta_i = \pi/2 + \pi/2i$.
At time $0$, place point $p_i$ at a distance of $1$ from the origin at an angle of $\theta_i$ from the positive $x$ direction.
Each point moves with speed $1$ in the direction $\theta_i - \pi/2$ (see Fig. \ref{fig:noLinesSame}).

\begin{figure}[h]
\center{
\includegraphics[height=2in,width=2.5in]{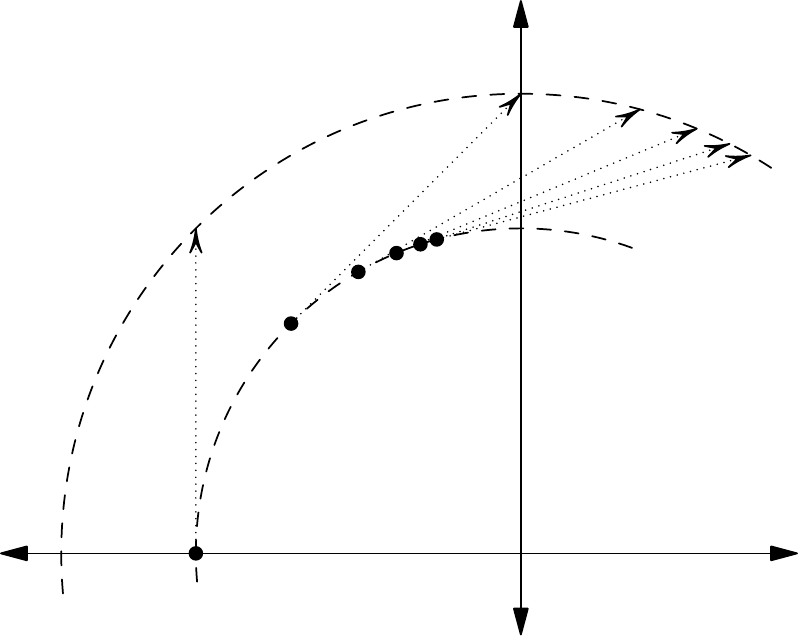}}
\caption{\label{fig:noLinesSame} A set of points, each moving at speed $1$, of which no three are ever collinear.}
\end{figure}

By this construction, the lines $L(p_i)$ will be from one ruling of a hyperboloid of a single sheet $S$ \cite{hilbert1952geometry}.
The intersection of any horizontal plane with $S$ will be a circle.
Since no line intersects a circle in more than two points, there will never be three points on any line.

In order to modify this construction so that no two points have the same speed, we will stretch it in the $x$-direction.

For $1 \leq i \leq n$, if $p_i$ is at location $(x_i, y_i)$ at time $0$, then place point $p'_i$ at location $(2x_i, y_i)$.
If the velocity vector of $p_i$ is $(v_{(x,i)}, v_{(y,i)})$, then the velocity vector of $p'_i$ is $(2v_{(x,i)}, v_{(y,i)})$ (see Fig. \ref{fig:noLinesDifferent}).

\begin{figure}[h]
\center{
\includegraphics[height=1.2in,width=2.8in]{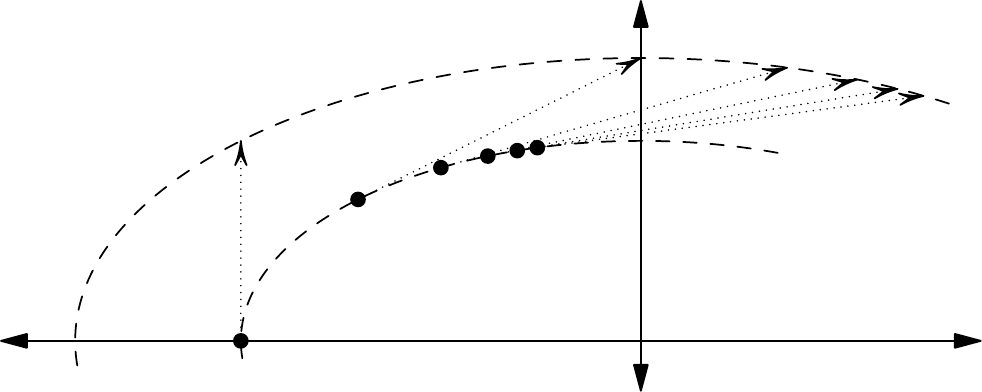}}
\caption{\label{fig:noLinesDifferent} A set of points, no two moving at the same speed, of which no three are ever collinear.}
\end{figure}

Since no two points $p_i, p_j \in \{p_1, p_2, ...,p_n\}$ have the same $x$ component to the vector describing their motion, no two points $p'_i, p'_j \in \{p_1, p_2, ...,p_n\}$ have the same speed.

The lines $L_{p'_i}$ are from one ruling of a hyperboloid of a single sheet $S'$.
The main difference between $S$ and $S'$ is that $S'$ is stretched in the $x$-direction, so the intersection of any horizontal plane with $S'$ is an ellipse rather than a circle.
No line intersects an ellipse in more than two points, so again there will never be three points on any line.
\end{proof}

\section{Conclusion}\label{sec:conclude}

We derived tight bounds on the minimum and maximum number of 3- and 4-collinearities among $n$ kinetic points, each moving with constant velocity in the plane. Our initial study poses more questions than it answers.

\begin{open}
What is the maximum number of $k$-collinearities among $n$ kinetic points in the plane?
Is our lower bound $\Omega(n^3/k^4 + n^2/k^2)$ tight?
\end{open}

\begin{open}
What is the maximum number of $k$-collinearities among $n$ kinetic points in the plane if no three points are always collinear and no two points collide?
\end{open}

\begin{open}
What is the maximum number of $3$-collinearities among $n$ kinetic points in the plane if the trajectory of each point is an algebraic curve of degree bounded by a constant $b$?
\end{open}

\begin{open}
A $d$-collinearity in $\mathbb{R}^d$ is called {\em full-dimensional} if not all points involved in the collinearity are in a hyperplane at all times. What is the maximum number of full-dimensional $d$-collinearities among $n$ kinetic points in $\mathbb{R}^d$?
\end{open}

The trajectories of $n$ kinetic points in $\mathbb{R}^d$ is an arrangement of $n$ nonhorizontal lines in $\mathbb{R}^{d+1}$. Recall that a $k$-collinearity corresponds to a {\em horizontal} line that intersects $k$ trajectories. If we drop the restriction to horizontal lines, we are led to the following problem.

\begin{open}
For an arrangement ${\mathcal A}$ of $n$ lines in $\mathbb{R}^3$, what is the maximum number of lines $L$ such that $L$ intersects at least $3$ lines in ${\mathcal A}$, which are not all concurrent and not all from a single ruling of a doubly ruled surface?
\end{open}

\small
\bibliographystyle{abbrv}
\bibliography{col_kinetic}

\begin{thebibliography}{10}

\bibitem{Basch99}
J.~Basch, L.~J. Guibas, and J.~Hershberger.
\newblock Data structures for mobile data.
\newblock {\em Journal of Algorithms}, 31:1--–28, 1999.

\bibitem{berger1987geometry}
M.~Berger.
\newblock {\em Geometry. {II}}.
\newblock Universitext. Springer-Verlag, Berlin, 1987.
\newblock Translated from the French by M. Cole~and S.~Levy.

\bibitem{ElekesKS11}
G.~Elekes, H.~Kaplan, and M.~Sharir.
\newblock On lines, joints, and incidences in three dimensions.
\newblock {\em J. Comb. Theory, Ser. A}, 118(3):962--977, 2011.

\bibitem{FeldmanS05}
S.~Feldman and M.~Sharir.
\newblock An improved bound for joints in arrangements of lines in space.
\newblock {\em Discrete {\&} Computational Geometry}, 33(2):307--320, 2005.

\bibitem{guibas04}
L.~Guibas.
\newblock Kinetic data structures.
\newblock In D.~Mehta and S.~Sahni, editors, {\em Handbook of Data Structures
  and Applications}, pages 23--1--23--18. Chapman and Hall/CRC, 2004.

\bibitem{Guth10}
L.~Guth and N.~H. Katz.
\newblock Algebraic methods in discrete analogues of the kakeya problem.
\newblock {\em Adv. in Math.}, 225:2828--2839, 2010.

\bibitem{hilbert1952geometry}
D.~Hilbert and S.~Cohn-Vossen.
\newblock {\em Geometry and the imagination}.
\newblock Chelsea Publishing Company, New York, NY, 1952.
\newblock Translated by P. Nem{\'e}nyi.

\bibitem{KaplanRS11}
H.~Kaplan, N.~Rubin, and M.~Sharir.
\newblock A kinetic triangulation scheme for moving points in the plane.
\newblock {\em Comput. Geom.}, 44(4):191--205, 2011.

\bibitem{KaplanSS10}
H.~Kaplan, M.~Sharir, and E.~Shustin.
\newblock On lines and joints.
\newblock {\em Discrete {\&} Computational Geometry}, 44(4):838--843, 2010.

\bibitem{KirkpatrickS02}
D.~G. Kirkpatrick and B.~Speckmann.
\newblock Kinetic maintenance of context-sensitive hierarchical representations
  for disjoint simple polygons.
\newblock In {\em Sympos. on Comput. Geom.}, pages 179--188. ACM Press, 2002.

\bibitem{Sharir94}
M.~Sharir.
\newblock On joints in arrangements of lines in space and related problems.
\newblock {\em J. Comb. Theory, Ser. A}, 67(1):89--99, 1994.

\bibitem{Szemeredi83}
E.~Szemer{\'e}di and W.~Trotter.
\newblock Extremal problems in discrete geometry.
\newblock {\em Combinatorica}, 3(3):381--392, 1983.

\end{thebibliography}

\end{document}